\documentclass[twocolumn,10pt]{IEEEtran}

\usepackage{braket}

\renewcommand{\baselinestretch}{0.88}

\input{def.tex}
\usepackage{dsfont}
\usepackage{minibox}
\DeclareSymbolFont{matha}{OML}{txmi}{m}{it}
\DeclareMathSymbol{\varv}{\mathord}{matha}{118}
\usepackage{eurosym}
\usepackage{hhline,physics}
\usepackage{algorithm,algorithmicx}
\usepackage{multicol}\IEEEoverridecommandlockouts
\begin{document}
	\title{ Coverage Probability of Distributed IRS Systems Under Spatially Correlated Channels}
	\author{Anastasios Papazafeiropoulos, Cunhua Pan, Ahmet Elbir, Pandelis Kourtessis, Symeon Chatzinotas, John M. Senior \thanks{A. Papazafeiropoulos is with the Communications and Intelligent Systems Research Group, University of Hertfordshire, Hatfield AL10 9AB, U. K., and with SnT at the University of Luxembourg, Luxembourg. C. Pan is with the School of Electronic Engineering and Computer Science at Queen Mary University of London, London E1 4NS, U.K. A. Elbir is with the EE department of Duzce University, Duzce, Turkey.
			P. Kourtessis and John M. Senior are with the Communications and Intelligent Systems Research Group, University of Hertfordshire, Hatfield AL10 9AB, U. K. S. Chatzinotas is with the SnT at the University of Luxembourg, Luxembourg. E-mails: tapapazaf@gmail.com, c.pan@qmul.ac.uk, ahmetmelbir@gmail.com \{p.kourtessis,j.m.senior\}@herts.ac.uk, symeon.chatzinotas@uni.lu.}}
	\maketitle

	\pagestyle{empty}

	\vspace{-2.9cm}
	
	\begin{abstract}
		This paper suggests the use of multiple distributed intelligent reflecting surfaces (IRSs) towards a smarter control of the propagation environment. Notably, we also take into account the inevitable correlated Rayleigh fading in IRS-assisted systems. In particular, in a single‐input and single‐output (SISO) system, we consider and compare two insightful scenarios, namely, a finite number of large IRSs and a large number of finite size IRSs to show which implementation method is more advantageous. In this direction, we derive the coverage probability in closed-form for both cases contingent on statistical channel state information (CSI) by using the deterministic equivalent (DE) analysis. Next, we obtain the optimal coverage probability. Among others, numerical results reveal that the addition of more surfaces outperforms the design scheme of adding more elements per surface. Moreover, in the case of uncorrelated Rayleigh fading, statistical CSI-based IRS systems do not allow the optimization of the coverage probability.
	\end{abstract}
	
	\begin{keywords}
		Intelligent reflecting surface (IRS), coverage probability, deterministic equivalents, beyond 5G networks.
	\end{keywords}
	\section{Introduction}
	The advancements in metasurfaces have enabled the development of intelligent reflecting surface (IRS), being a planar array that includes a large number of nearly passive reflecting elements \cite{Wu2020,Basar2019}. IRS provides a smart radio environment by realizing reflecting beamforming through its elements, which can introduce phase adjustments on the impinging wave with different objectives such as an increase of coverage and avoidance of obstacles. Also, its construction principles allow affordable and green transmission due to its low-cost hardware and low energy consumption, respectively.
	
	Many works have approached the concept of IRS from the wireless communication point of view due to its appealing advantages to achieve various tasks by adjusting the phase shifts of the reflecting surface elements, e.g., see \cite{Basar2019,Huang2019a,Bjoernson2019b,Pan2020,Kammoun2020, Elbir2020,Najafi2020,Guo2020} and references therein. Among them, in \cite{Pan2020}, we observe a maximization of the sum-rate with a transmit power constraint, in \cite{Huang2019a}, authors achieved maximization of the energy efficiency with signal-to-interference-plus-noise ratio (SINR) constraints, and in \cite{Elbir2020}, channel estimation, being an interesting research area in IRS-assisted systems due to their special characteristics,  was proposed by using a deep learning approach.
	
	In particular, in the communication-theoretic direction, the study of the coverage probability in IRS-assisted systems has attracted significant attention \cite{Guo2020,Yang2020}. However, all previous works assumed only one IRS while the performance of systems aided simultaneously by multiple IRSs, offering extended advantages such as a more robust avoidance of obstacles and improved coverage, has not been investigated except in \cite{Gao2020,Zhang2019a,Sun2020,Zhang2020}.
	
	In parallel, the assumption of independent Rayleigh fading, which is often assumed for tractable performance analysis, is unrealistic for IRS-assisted systems \cite{Bjoernson2020}. Although several works have accounted for the IRS correlation by acknowledging its importance, they relied on conventional correlation models \cite{Kammoun2020}, which are not directly applicable in IRSs as mentioned in \cite{Bjoernson2020}, where a practical correlation model was suggested.

Against the above background, we present the only work providing the coverage probability in closed-form for single‐input and single‐output (SISO) systems assisted simultaneously by multiple IRSs while accounting for the inevitable correlated Rayleigh fading requiring suitable mathematical manipulations. In particular, we consider two insightful design cornerstones: i) A finite number of large IRSs (the number of elements per IRS grows large) and ii) a large number of IRSs with each IRS having finite dimensions. Hence, contrary to \cite{Guo2020}, we establish the theoretical framework incorporating correlated fading into the analysis to identify the realistic potentials of IRSs before their final implementation and we also study the performance when the IRSs number becomes large. Compared to \cite{Gao2020}, which also assumed distributed IRSs and correlated fading, we focus on the coverage probability instead of the achievable rate and we rely on a more realistic correlation model while we account for the scenario of a large number of IRSs, which has not been addressed before. Moreover, we provide a methodology to optimize the reflect beamforming matrix based on statistical channel state information (CSI) that enables optimization at every several coherence intervals instead of frequent optimization at every coherence interval as in works relying on instantaneous CSI.

	\textit{Notation}: Vectors and matrices are denoted by boldface lower and upper case symbols, respectively. The notations $(\cdot)^\T$, $(\cdot)^\H$, and $\tr\!\left( {\cdot} \right)$ represent the transpose, Hermitian transpose, and trace operators, respectively. The expectation operator is denoted by $\EE\left[\cdot\right]$ while $ \diag\left(\ba \right) $ expresses a diagonal matrix with diagonal elements being the elements of vector $ \ba $ and $\diag\left(\bA\right) $ expresses a vector with elements the diagonal elements of $ \bA $. Also, the notations $ \arg\left(\cdot\right) $ and $ \mod(\cdot,\cdot) $ denote the argument function and the modulus operation while $ \lfloor \cdot \rfloor $ truncates the argument. Given two infinite sequences $a_n$ and $b_n$, the relation $a_n\asymp b_n$ is equivalent to $a_n - b_n \xrightarrow[ n \rightarrow \infty]{\mbox{a.s.}} 0$. Finally, $\bb \sim \cC\cN{(\b0,\mathbf{\Sigma})}$ represents a circularly symmetric complex Gaussian vector with {zero mean} and covariance matrix $\mathbf{\Sigma}$.
	
	\section{System Model}\label{System}
	We consider the smart connectivity between a single-antenna transmitter (TX) and a single-antenna receiver (RX) enabled by means of a set of $ M $ independent IRSs uniformly distributed in the intermediate space,  which are placed in the locations of obstacles. In other words, we assume multiple blocked areas, and the IRSs are expected to cover these areas. Also, we assume that each IRS, controlled through a perfect backhaul link by the transmitter, consists of a two-dimensional rectangular grid of $ N =N_{\mathrm{H}}N_{\mathrm{V}}$ passive unit elements with $ N_{\mathrm{H}} $ elements per row and $ N_{\mathrm{V}} $ elements per column that can modify the phase shifts of impinging waves. Severe blockage effects make any direct channel unavailable. To focus on the impact of the multitude of IRSs, we rely on the assumption of perfect CSI, and thus, the results act as upper bounds of practical implementations.
	
	Let a block-fading model with independent realizations across different coherence blocks for the description of all channels. In particular, we assume the existence of a direct link and $ M $ cascaded channels. The former is described by $ h_{\mathrm{d}}\sim \mathcal{CN}\left(0,\beta_{\mathrm{d}}\right) $, where $ \beta_{\mathrm{d}} $ is the path-loss. Regarding the IRS-assisted links, $ 	\bh_{m,1}=\left[h_{mn,1}, \ldots, h_{mN,1}\right]^{\T}\in \mathbb{C}^{N \times 1}$ expresses the channel fading vector between the TX and the $ m $th IRS while $ \bh_{m,2}=\left[h_{mn,2}, \ldots, h_{mN,2}\right]^{\T} \in \mathbb{C}^{N \times 1}$ corresponds to the link between the $ m $th IRS and the RX. Contrary to existing works, relying on independent Rayleigh and Rician fading models, we consider correlated Rayleigh fading.\footnote{The extension to correlated Rician fading, having a LoS component, is the topic of future work. Also, in this work we have considered a rich scattering environment while there are also works assuming a limited number of scatterers such as \cite{Najafi2020}.} Hence, by accounting for both small-scale fading and path-loss, we have
	\begin{align}
		\bh_{m,1}&\sim \mathcal{CN}\left(\b0, \beta_{m,1}\bR_{m,1}\right),\\
		\bh_{m,2}&\sim \mathcal{CN}\left(\b0, \beta_{m,2}\bR_{m,2}\right),
	\end{align}
	where $ \beta_{m,1} $, $ \beta_{m,2} $ are the path-losses while $ \bR_{m,1}\in \mathbb{C}^{N \times N} $, $ \bR_{m,2} \in \mathbb{C}^{N \times N}$ is the spatial covariance matrices of the respective links. \footnote{The path-losses and the covariance matrices are assumed known by applying practical methods, {e.g., see} \cite{Neumann2018}.} Herein, we use the correlation model proposed in \cite{Bjoernson2020} as suitable for IRSs under the conditions of rectangular IRSs and isotropic Rayleigh fading. Let the size of each IRS element be $ d_{\mathrm{H}}\times d_{\mathrm{V}} $, where $d_\mathrm{V}$ and $d_\mathrm{H}$ express its vertical height and its horizontal width, respectively. Then, the $ \left(i,j\right) $th element of the correlation matrix $ \bR_{m,k}, $ $ k\!\in \!\{1,2\} $ is given by
	\begin{align} \label{eq:Element}
		r_{ij,mk} = d_{\mathrm{H}} d_{\mathrm{V}} \mathrm{sinc} \left( 2 \|\mathbf{u}_{i,mk} - \mathbf{u}_{j,mk} \|/\lambda\right),
	\end{align}
	where $\mathbf{u}_{\epsilon,mk} = [0, \mod(\epsilon-1, N_\mathrm{H})d_\mathrm{H}, \lfloor (\epsilon-1)/N_\mathrm{H} \rfloor d_\mathrm{V}]^\T$, $\epsilon \in \{i,j\}$ and $\lambda$ is the wavelength of the plane wave.
	
	Based on a slowly varying flat-fading channel model, the complex-valued received signal at the RX through the network of $ M $ IRSs is described by
	\begin{align}
		y=\bigg(\sum^{M}_{m=1}\bh_{m,1}^{\H}\bPhi_{m}\bh_{m,2}+h_{\mathrm{d}}\!\bigg)x+n,\label{received}
	\end{align}
	where $ \bPhi_{m}=\mathrm{diag}\left(\al_{m1} \exp\left(j \theta_{m1}\right), \ldots, \al_{mN} \exp\left(j \theta_{mN}\right)\right)\in \mathbb{C}^{N\times N}$ expresses the response of the elements of the $ m $th IRS with $ \theta_{mn} \in \left[ 0, 2 \pi \right], n=1,\ldots,N$ and $ \al_{mn} \in (0,1]$ being the phase shifts and the fixed amplitude reflection coefficients of the corresponding IRS element. The progress on loss-less meta-surfaces allows to set $ \al_{mn}=1$, which ensures maximum reflection \cite{Bjoernson2019b}. Also, $ n\sim\mathcal{CN}\left(0,N_{0}\right) $ is the additive white Gaussian noise (AWGN) sample and $ x $ is the transmitted data symbol satisfying $ \EE[|x|^{2}] \!=\!P$, i.e., $ P $ denotes the average power of the symbol.

	\section{Performance Analysis}\label{Performance Analysis} 
	In this section, we present the derivation of the coverage probability when multiple IRSs are subject to correlated Rayleigh fading by means of the deterministic equivalent (DE) analysis, which provides tight approximation even for finite practical dimensions (see \cite{Papazafeiropoulos2015a,Papazafeiropoulos2016} and references therein). We focus on two interesting scenarios: a) a finite set of large IRSs ($ N \to \infty $); and b) a large number of IRSs ($ M \to \infty $) with finite size.
	
	
	\subsection{Main Results}	
	The coverage probability $ \bar{P}_{\mathrm{c}} $ is defined as the probability that the effective received SNR at the RX is larger than a given threshold $ T $, i.e., $ \bar{P}_{\mathrm{c}}=\mathrm{Pr}\left(\gamma>T\right) $,
	where 
	\begin{align}
		\gamma ={\gamma_{0}}{\bigg|\displaystyle\sum^{M}_{m=1}\bh_{m,1}^{\H}\bPhi_{m} \bh_{m,2}+h_{\mathrm{d}}\bigg|^{2}}{}\label{general}
	\end{align}
	is the received SNR in the general case with correlated fading that is obtained by using \eqref{received} and assuming coherent communication. Also, $ {\gamma_{0}}=P/N_{0} $ is the average transmit SNR. Under independent Rayleigh fading with instantaneous CSI, it is known that the phase configuration $ \phi_{m,n}= \arg\left(h_{\mathrm{d}}\right)-\arg\left(h_{mn,1}^{*}\right)\arg\left(h_{mn,2}\right)$ provides the optimal SNR \cite{Basar2019,Bjoernson2019b}. However, in the practical case of correlated Rayleigh fading, where only statistical CSI is available, we cannot directly obtain the solution of the phase shifts. Moreover, since correlated fading renders the exact derivation of the SNR intractable, we resort to the application of the DE analysis to derive the approximated SNR. \footnote{Note that the majority of works, deriving the coverage probability in IRS-assisted systems, result in approximations since they are based on CLT.} In Section \ref{Numerical}, we show that the corresponding $ P_{\mathrm{c}} $ provides a tight match with $ \mathrm{Pr}\left(	\gamma>T\right) $.
	\subsection{Finite $ M $ and large $ N $ analysis}
	In this part, we assume large IRSs, as usually considered in the existing literature to obtain the coverage probability, e.g., see \cite{Yang2020,Zhang2019a}.
	\begin{lemma}\label{DE1}
		The SNR of a SISO transmission, enabled by $ M $ large IRSs with correlated Rayleigh fading is approximated by
		\begin{align}
			\!{\gamma}
			&\asymp \gamma_{0}\!\left(B_{M}+|h_{\mathrm{d}}|^{2}\right)\!,\!\label{DE_SNR1}
		\end{align}
		where $ B_{M}\!=\!\sum^{M}_{m=1}\!\beta_{m}\!\tr\!\left( \bR_{m,1}\bPhi_{m}\bR_{m,2}\bPhi_{m}^{\H}\right) $ with $ \beta_{m}\!=\!\beta_{m,1} $ $\beta_{m,2} $.
	\end{lemma}
	\begin{proof}
		The proof starts by dividing \eqref{general} with $ \frac{1}{N^{2}} $. Then, we have
		\begin{align}
			\!\!\!	&\!\!	\frac{1}{N^{2}}{\gamma}=\gamma_{0}\frac{1}{N^{2}}\bigg(\displaystyle\sum^{M}_{m=1}\bigg|\bh_{m,1}^{\H}\bPhi_{m} \bh_{m,2}\bigg|^{2}+|h_{\mathrm{d}}|^{2}\nn\\&\!\!
			+2\mathrm{Re}\left(h_{\mathrm{d}}^{*}\displaystyle\sum^{M}_{m=1}\bh_{m,1}^{\H}\bPhi_{m} \bh_{m,2}\right)\nn\\
			\!\!\!	&\!\!+\displaystyle\sum^{M}_{m=1}\sum^{M}_{\substack{j=1 \\ n\ne m}}\bh_{m,1}^{\H}\bPhi_{m} \bh_{m,2}\bh_{n,2}^{\H}\bPhi_{n}^{\H} \bh_{n,1}\bigg)
			\label{DE_SNR10}\\
			\!\!\!
			&\!\!\asymp \gamma_{0}\frac{1}{N^{2}}\!\!\left(\displaystyle\sum^{M}_{m=1}\!\beta_{m,1}\beta_{m,2}\!\tr\!\left( \bR_{m,1}\bPhi_{m}\bR_{m,2}\bPhi_{m}^{\H}\right)\!+\!|h_{\mathrm{d}}|^{2}\!\!\right)\!\!,\label{DE_SNR3}
		\end{align}
		where, in \eqref{DE_SNR3}, we have used \cite[Lem. 4]{Papazafeiropoulos2015a} for the first, third, and the fourth terms. Especially, the third and fourth terms in \eqref{DE_SNR10} vanish as $ N\to \infty $ due to the independence among the $ M $ IRSs and between the two links, respectively. 
	\end{proof}
	\begin{proposition}\label{GeneralPDF}
		The coverage probability of a SISO transmission, enabled by $ M $ large IRSs with correlated Rayleigh fading for arbitrary phase shifts, is tightly approximated by
		\begin{align}
			\!	\!	P_{\mathrm{c}}=\left\{
			\begin{array}{ll}\mathrm{exp}\!\left(\!-\frac{1}{\beta_{\mathrm{d}}}\!\!\left(\!\frac{T}{\gamma_{0}}-B_{M}\!\right)\!\!\right)&B_{M}<\frac{T}{\gamma_{0}}\\
				1& B_{M}\ge \frac{T}{\gamma_{0}}
			\end{array} 
			\right. 		.\label{PC11}
		\end{align}
		
	\end{proposition}
	\begin{proof}
		The coverage probability is written as
		\begin{align}
			&P_{\mathrm{c}}\!=\!\mathrm{Pr}\!\left(\!|h_{\mathrm{d}}|^{2}>\frac{T}{\gamma_{0}}\!-\!B_{M}\!\!\right)\!,\label{PC1}
		\end{align}
		where in \eqref{PC1}, we have used the SNR from \eqref{DE_SNR1}. Given that $ |h_{\mathrm{d}}|^{2} $ is exponentially distributed with rate parameter $ \beta_{\mathrm{d}} $, i.e., $ |h_{\mathrm{d}}|^{2}\sim \mathrm{Exp}(1/\beta_{\mathrm{d}}) $, we obtain the first branch in \eqref{PC11}, if $ B_{M}<\frac{T}{\gamma_{0}} $. Otherwise, $ P_{\mathrm{c}} =1$, and we conclude the proof.
	\end{proof}
	\begin{remark}
		If the aggregate contribution from the IRS-assisted channels is larger than $ T/\gamma_{0} $, no outage is detected during the communication. Also, the weaker the direct signal $ (\beta_{\mathrm{d}}\to 0) $, the less severe its impact is on $ 	P_{\mathrm{c}} $ and the influence of the cascaded channels becomes more pronounced. Moreover, when the path-losses of the cascaded links increase, i.e., $ \beta_{m} $ decreases, the coverage decreases.
	\end{remark}
	\begin{remark}
		From \eqref{PC11}, we observe that when the number of surfaces $ M $ increases, the coverage probability is improved. In addition, by increasing the size of each IRS in terms of $ N $, 	$ P_{\mathrm{c}} $ is enhanced. Hence, the use of more IRSs or larger IRSs is proved to be beneficial for coverage. 
	\end{remark}
	\begin{remark}
		Obviously, the coverage probability depends on the phase shifts, which could be optimized. However, in the case of uncorrelated fading, i.e., $ \bR_{m,1}=\bR_{m,2}=\Id_{N} $, $ P_{\mathrm{c}} $ becomes independent of the reflect beamforming matrices $ \bPhi_{m} $. In such a case, the phase shifts of the IRS cannot be optimized to improve the coverage. 
	\end{remark}
	
	\subsection{Large $ M $ and finite $ N $ analysis}
	The previous analysis does not allow to examine the coverage when $ M \to \infty $ but $ N $ is finite. To address this scenario, let $ 	\bg_{n,1}=\left[h_{1n,1}, \ldots, h_{Mn,1}\right]^{\T}\in \mathbb{C}^{M \times 1}$ denote the channel fading vector between the TX and the $ n $th elements of all IRSs (first link). Also, $ \bg_{n,2}=\left[h_{1n,2}, \ldots, h_{Mn,2}\right]^{\T} \in \mathbb{C}^{N \times 1}$ expresses the channel between the $ n $th elements of all IRSs and the RX (second link). Given that IRSs are reasonably far apart each other, we assume no correlation among them.\footnote{Note that not only this assumption is quite reasonable but the modeling of a potential correlation would require the conduct of measurement campaigns, which are currently unavailable.} In other words, we have $ \EE[\bh_{m,1} \bh_{l,1}^{\H}]=\EE[\bh_{m,2} \bh_{l,2}^{\H}]=\b0_{N}~\forall~m\ne l$ with $ m=1,\ldots,M $ and $ l=1,\ldots,M $. 
	
	Notably, a correlation appears between different channel vectors at each link. Specifically, regarding the first link, let $ \bQ_{np,1} $ describe the correlation between the $ n $th and $ p $th elements across all surfaces. It can be written as
	\begin{align}
		\bQ_{np,1}&=\EE[\bg_{n,1}\bg_{p,1}^{\H}] \nn\\
		&=\betav_{1}\diag\left(r_{np,1}^{1},\ldots,r_{np,M}^{1}\right),
	\end{align}
	where the matrix $ \betav_{1} =\diag\left(\beta_{1,1}, \ldots, \beta_{M,1}\right)\in \mathbb{C}^{M \times M}$ is diagonal with elements expressing the path-losses between the TX and the $ M $ surfaces. Note that $ \betav_{1} $ does not depend the index $ n $ but it includes the corresponding path-losses from all IRSs. The matrix $ \bQ_{np,1} $ is diagonal due to the independence among the IRSs. Also, $ r_{np,i}^{1}$ with $ i=1,\ldots,M $ expresses the $ (n,p) $th element of the correlation matrix of the $ i $th IRS of the first link, i.e., $ \bR_{i,1} $. Similarly, for the second link, we have 
	\begin{align}
		\bQ_{np,2}&=\betav_{2}\diag\left(r_{np,1}^{2},\ldots,r_{np,M}^{2}\right),
	\end{align}
	where $ \betav_{2}=\diag\left(\beta_{1,1}, \ldots, \beta_{M,1}\right) \in \mathbb{C}^{M \times M} $ is the diagonal matrix expressing the path-losses among the IRSs and the RX, and $ \bQ_{np,2} =\diag\left(r_{np,1}^{2}, \ldots, r_{np,M}^{2}\right)\in \mathbb{C}^{M \times M} $ describes the corresponding spatial correlation. Notably, in the case of independent Rayleigh fading, $ \bQ_{np,1}=\bQ_{np,2}=\bO $ for $ n\ne p $. As a result, the corresponding channel vectors of the first and and second links are formulated as
	\begin{align}
		\bg_{n,1}&\sim \mathcal{CN}\left(\b0, \betav_{1}\bQ_{nn,1}\right),\\
		\bg_{n,2}&\sim \mathcal{CN}\left(\b0, \betav_{2}\bQ_{nn,2}\right).
	\end{align}

	The SNR in \eqref{general} can be rewritten in terms of a summation over the number of elements of each IRS as
	\begin{align}
		\gamma ={\gamma_{0}}{\bigg|\displaystyle\sum^{N}_{n=1}\bg_{n,1}^{\H}\bPsi_{n} \bg_{n,2}+h_{\mathrm{d}}\bigg|^{2}}{},\label{generalM1}
	\end{align}
	where $ \bPsi_{n}=\mathrm{diag}\left(\exp\left(j \theta_{1n}\right), \ldots, \exp\left(j \theta_{Mn}\right)\right)\in \mathbb{C}^{M\times M} $.
	
	\begin{lemma}\label{GeneralPDFM}
		The SNR of a SISO transmission, enabled by a large number of finite size IRSs with correlated Rayleigh fading is approximated by
		\begin{align}
			{\gamma}
			&\asymp {\gamma_{0}}\left(B_{N}+|h_{\mathrm{d}}|^{2}\right)\!,\label{DE_SNR2}
		\end{align}
		where $ B_{N}= \sum^{N}_{n=1}\!\sum^{N}_{p=1}\!\tr\!\left(\bQ_{np,1}\bPsi_{n} \bQ_{np,2}\bPsi_{p}^{\H}\right)$. Note that $ \bPsi=\diag(\bPsi_{1},$ $\ldots, \bPsi_{N})\in \mathbb{C}^{MN \times MN} $, i.e., $ \bPsi $ is a block diagonal matrix.
	\end{lemma}
	\begin{proof}
		By dividing $ 	\gamma $ with $ \frac{1}{M^{2}} $, we have
		\begin{align}
			&	\!\!\!\frac{1}{M^{2}}	\gamma\!=\!{\gamma_{0}}\frac{1}{M^{2}}\!\bigg(\displaystyle\!\sum^{N}_{n=1}\!\bigg|\bg_{n,1}^{\H}\bPsi_{n} \bg_{n,2}\bigg|^{2}\!\!\!+\!2\mathrm{Re}\!\left(\!\!h_{\mathrm{d}}^{*}\!\displaystyle\sum^{N}_{n=1}\!\bg_{n,1}^{\H}\bPsi_{n} \bg_{n,2}\!\!\right)\nn\\
			&+|h_{\mathrm{d}}|^{2}+\displaystyle\sum^{N}_{n=1}\sum^{N}_{\substack{p=1 \\ p\ne n}}\bg_{n,1}^{\H}\bPsi_{n} \bg_{n,2}\bg_{p,2}^{\H}\bPsi_{p}^{\H}\bg_{p,1}\bigg)\label{generalM2}\\
			&={\gamma_{0}}\frac{1}{M^{2}}\bigg(\displaystyle\sum^{N}_{n=1}\sum^{N}_{p=1}\tr\left(\bQ_{np,1}\bPsi_{n} \bQ_{np,2}\bPsi_{p}^{\H}\right)+|h_{\mathrm{d}}|^{2}\bigg)\label{generalM4},
		\end{align}
		where the second term in \eqref{generalM2} vanishes as $ M\to \infty $ due to the independence between the two links. Application of \cite[Lem. 4]{Papazafeiropoulos2015a} at the first and fourth terms in \eqref{generalM2} gives 
		\eqref{generalM4} after a direct combination of the two resultant traces.
	\end{proof}
	\begin{proposition}\label{GeneralPDF1}
		The coverage probability of a SISO transmission, enabled by a large number of finite size IRSs with correlated Rayleigh fading for arbitrary phase shifts, is tightly approximated by
		\begin{align}
			\!	\!	P_{\mathrm{c}}=\left\{
			\begin{array}{ll}\mathrm{exp}\!\left(\!-\frac{1}{\beta_{\mathrm{d}}}\!\!\left(\!\frac{T}{\gamma_{0}}-B_{N}\!\right)\!\!\right)&B_{N}<\frac{T}{\gamma_{0}}\\
				1& B_{N}\ge \frac{T}{\gamma_{0}}
			\end{array} 
			\right. 		
			\label{PC13}\!.
		\end{align}
	\end{proposition}
	\begin{proof}
		The proof follows similar lines with the proof of Proposition \ref{GeneralPDF}.
	\end{proof}
	\begin{remark}
		We observe that \eqref{PC13} has a similar expression with \eqref{PC11}. However, the main characteristic of \eqref{PC13} is that it is written in terms of a double summation expressing the correlation among the IRS elements instead of one summation in \eqref{PC11}. Notably, if the correlation matrix is identical across different IRSs, $ 	\bQ_{np,1} $ and $ 	\bQ_{np,2} $ are scaled identity matrices but the coverage will always be dependent on the phase shifts due to the contributions from the off-diagonal terms of the IRSs. Furthermore, we observe a similar dependence from the path-loss of the direct signal and the number of surfaces and their elements, i.e., their increase improves the coverage. Also, under uncorrelated Rayleigh fading conditions, $ P_{\mathrm{c}} $ does not depend on the phases, and thus, cannot be optimized.
	\end{remark}
	
	\subsection{Reflecting beamforming optimization}
	Both Propositions \eqref{GeneralPDF} and \eqref{GeneralPDF1} are described by a similar expression in terms of a trace that includes the reflecting beamforming matrices. 
	Hence, their optimization follows similar steps up to a point. Specifically, to achieve maximum $ P_{\mathrm{c}} $, we formulate the optimization problem, relying on the common assumption of infinite resolution phase shifters, as
	\begin{align}\begin{split}
			&\!\!\!(\mathcal{P}1)~\max_{\bPhi} ~~	P_{\mathrm{c}}\\
			&~~~~~~\mathrm{s.t}~|\phi_{mn}|\!=\!1,~~ m\!=\!1,\dots,M~\mathrm{and}~n\!=\!1,\dots,N,
		\end{split}\label{Maximization} 
	\end{align}
	where $ 	P_{\mathrm{c}} $ is given by \eqref{PC11} or \eqref{PC13} and $ \phi_{mn}= \exp\left(j \theta_{mn}\right) $.

	The optimization problem $ 	(\mathcal{P}1) $ is non-convex with respect to the reflect beamforming matrix while having a unit-modulus constraint regarding $ \phi_{mn} $. Use of projected gradient 	ascent until converging to a 	stationary point can provide a direct solution. In particular, since each surface has a similar solution, we focus on the $ m $th IRS. At the $ i $th step, we assume the vectors $ \bs_{m,i} =[\phi_{m1}^{i}, \ldots, \phi_{mN}^{i}]^{\T}$, which include 	the phases at this step. The next iteration increases $ P_{\mathrm{c}} $ until its convergence by projecting the solution onto the closest feasible point based on $ \min_{|\phi_{mn} |=1, n=1,\ldots,N}\|\bs_{m}-\tilde{\bs}_{m}\|^{2} $ satisfying the unit-modulus constraint concerning $ \phi_{mn} $ with
	\begin{align}
		\tilde{\bs}_{m,i+1}&=\bs_{m,i}+\mu \bq_{m.i},\label{sol1}\\
		\bs_{m,i+1}&=\exp\left(j \arg \left(\tilde{\bs}_{m,i+1}\right)\right).\label{sol2}
	\end{align}
	Note that $ \mu $ expresses the step size computed at each iteration by means of the backtracking line search \cite{Boyd2004} while $ \bq_{m,i} $ denotes the adopted ascent direction at step $ i $ with $ \bq_{m,i}= \pdv{	P_{\mathrm{c}}}{\bs_{m,i}^{*}} $, obtained by Lemma \ref{deriv1} below. Algorithm \ref{Algoa1} provides an outline of the proposed algorithm for Proposition \ref{GeneralPDF} and \ref{GeneralPDF1} by setting $ \tilde{\bPhi}= \bPhi_{m}$ and $ \tilde{\bPhi}= \bPsi_{n} $, respectively.
	\begin{algorithm}
		\caption{Projected Gradient Ascent Algorithm for the IRS Design}
		1.				 \textbf{Initialisation}: $ \bs_{m,0} =\exp\left(j\pi/2\right)\one_{N}$, $ \tilde{\bPhi}_{0}=\diag\left(\bs_{m,0}\right) $, $ P_{\mathrm{c}}^{0}=f\left(\tilde{\bPhi}_{0}\right) $ given by \eqref{Maximization}; $ \epsilon>0 $\\
		2. \textbf{Iteration} $ i $: \textbf{for} $ i=0,1,\dots, $ do\\
		3. $[\bq_{m,i}]_{n}= \pdv{P_{\mathrm{c}}}{\bs_{m,i}^{*}}$, where $\pdv{	P_{\mathrm{c}}}{\bs_{m,i}^{*}} $ is given by Lemma \ref{deriv1};\\
		4. \textbf{Find} $ \mu $ by backtrack line search$( f\left(\tilde{\bPhi}_{0}\right),\bq_{m,i},\bs_{m,i})$ \cite{Boyd2004};\\
		5. $ \tilde{\bs}_{m,i+1}=\bs_{m,i}+\mu \bq_{m,i} $;\\
		6. 	$ \bs_{m,i+1}=\exp\left(j \arg \left(\tilde{\bs}_{m,i+1}\right)\right) $; $ \tilde{\bPhi}_{i+1}= \diag\left(\bs_{m,i+1}\right) $;\\
		7. $P_{\mathrm{c}}^{i+1}=f\left(\tilde{\bPhi}_{i+1}\right) $;\\
		8. \textbf{Until} $ \|P_{\mathrm{c}}^{i+1}- P_{\mathrm{c}}^{i}\|^{2} <\epsilon$; \textbf{Obtain} $ \tilde{\bPhi}^{*}=\tilde{\bPhi}_{i+1}$;\\
		9. \textbf{end for}\label{Algoa1}
	\end{algorithm}
	\begin{lemma}\label{deriv1}
		The derivative of the coverage probability with respect to $ \bs_{m,i}^{*}$ is given by
		\begin{align}
			\!\pdv{	{P_{\mathrm{c}}}}{\bs_{m,i}^{*}}\!\!=\!\frac{\beta_{m}}{\beta_{\mathrm{d}}}\!\left\{\!\!\!\!
			\begin{array}{ll}	{P_{\mathrm{c}}}\,\diag\left(\bR_{m,1}\bPhi_{m} \bR_{m,2}\right),&\mathrm{Prop.}~\ref{GeneralPDF}\\
				{P_{\mathrm{c}}}			\displaystyle\!\sum^{N}_{p=1} \bc_{p},& \mathrm{Prop.}~\ref{GeneralPDF1}
			\end{array} 
			\right.\!\! 
		\end{align}
		when $ B_{i}<\frac{T}{\gamma_{0}},~i=M,N $. Otherwise, it is zero. Note that $\bc_{p}= \left[r_{1p,m}^{1} r_{1p,m}^{2}\phi_{m,1n}, \ldots,r_{Np,m}^{1} r_{Np,m}^{2}\phi_{m,1N} \right]^{\T} $.
	\end{lemma}
	\begin{proof}
		See Appendix~\ref{ArbitraryPDFProof2}.
	\end{proof}

	\section{Numerical Results}\label{Numerical} 
	Relying on a Cartesian coordinate system, we consider a cell, where the TX is located at the origin and the RX at $ \left(60, 0\right) $. Also, we assume a number of $ M=15 $ obstacles, being uniformly distributed between them, and at each obstacle location, we place an IRS to improve coverage. Each IRS is deployed with $ N=225 $ elements unless otherwise specified. The size of each IRS element is given by $ d_{\mathrm{H}}\!=\!d_{\mathrm{V}}\!=\!\lambda/2 $ \cite{Bjoernson2020}. The spatial correlation matrix is given by \eqref{eq:Element}. The large-scale fading coefficients between the TX and the RX are given by
	$\beta_{m,i} = G_t + G_r + 10 \nu_e \log_{10} (d_i/1\mathrm{m}) - 27.5$,
	where $i \in \{ 1, 2\}$ while $ \beta_{\mathrm{d}} $ is given similarly. Also, the path-loss exponents are $\nu_1 =\nu_2= 2$ and $ \nu_{\mathrm{d}}=3.5 $. \footnote{The path-loss exponent mainly depends on the properties of the obstacles and their densities, i.e., the exponent is high when the density of obstacles is high. Since, in IRS-aided systems, the IRS may be deployed in a proper place with a less density of obstacles, it is reasonable to set the path-loss exponent related to the IRS channels as small, while setting that of the direct channel as high.} Moreover, we have $G_t = 3.2$~dBi and $G_r = 1.3$~dBi. The system bandwidth is $10$~MHz, the carrier frequency is $3$~GHz, and the noise variance is $-94$~dBm with the noise figure being~$10$~dB. Note that the transmitter power is $10$~dBm. Monte-Carlo (MC) simulations verify the analytical results and corroborate that the DE analysis provides tight approximations as has been already shown in the literature, e.g., see \cite{Papazafeiropoulos2015a,Papazafeiropoulos2016}.
	
	Fig. \ref{Fig1}.(a) shows the coverage probability versus the target rate in the scenario described by Proposition \ref{GeneralPDF}.\footnote{The theoretical analysis of this proposition relies on finite $N $ but we consider $ N\ge 225 $, which is common for practical IRS implementations.} By increasing the number of elements in each IRS, $ P_{\mathrm{c}} $ increases. Also, the addition of more IRSs (an increase of $ M $) through $ B_{M} $ contributes to the observation of less outage during communication. Moreover, if no correlation is assumed, $ P_{\mathrm{c}} $ is lower because it becomes independent of the reflect beamforming matrix and cannot be optimized. Notably, if the impact from the direct signal through $ \beta_{\mathrm{d}} $ becomes weaker, the coverage decreases, and the variation regarding the number of elements as well as the correlated Rayleigh fading have a greater impact on $ P_{\mathrm{c}} $ since the relevant gaps are larger.
	\begin{figure*}[t]
		\begin{minipage}{0.33\textwidth}
			\centering
			\includegraphics[trim=0cm -0.20cm 0cm 0.2cm, clip=true, width=2.2in]{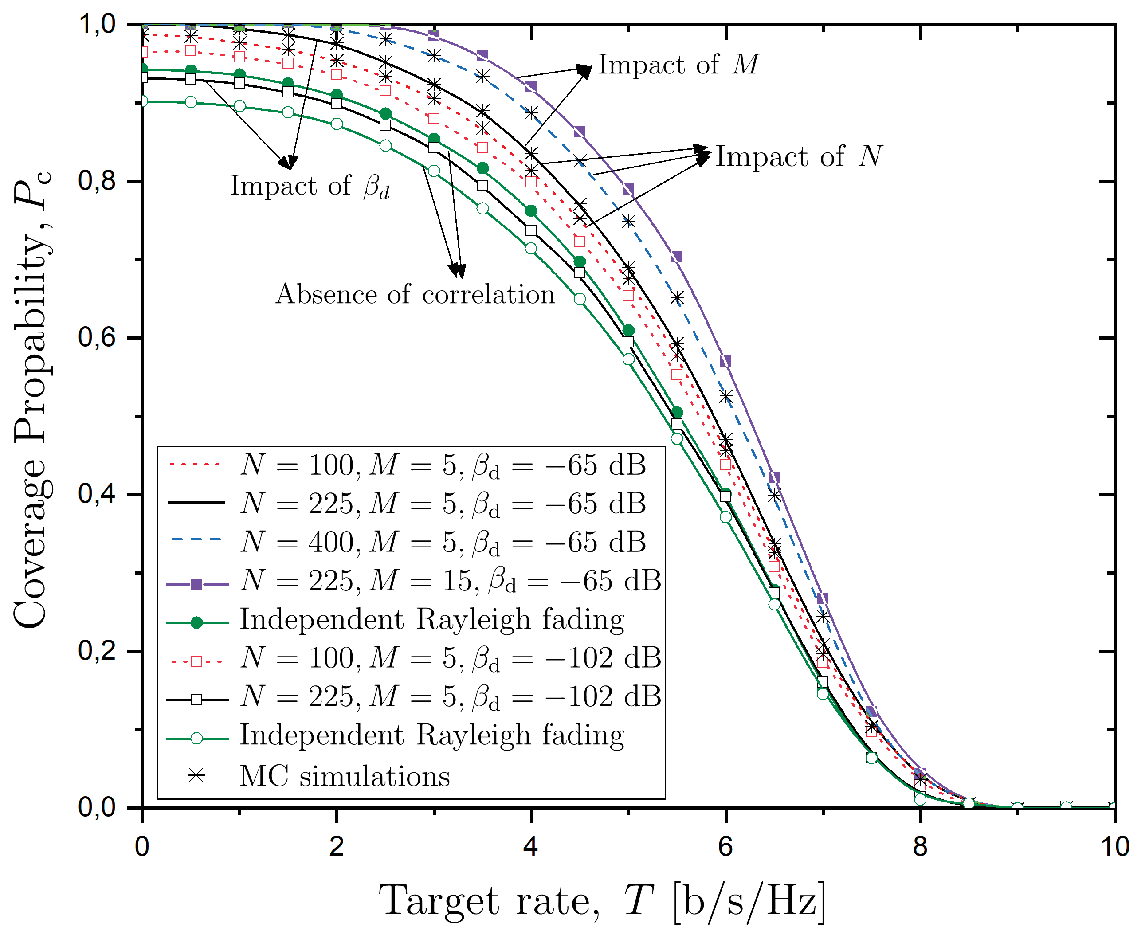} \vspace*{-0.2cm}
			\\ $(a)$
			\vspace*{-0.2cm}
		\end{minipage}
		\begin{minipage}{0.33\textwidth}
			\centering
			\includegraphics[trim=0cm -0.20cm 0cm 0.2cm, clip=true, width=2.2in]{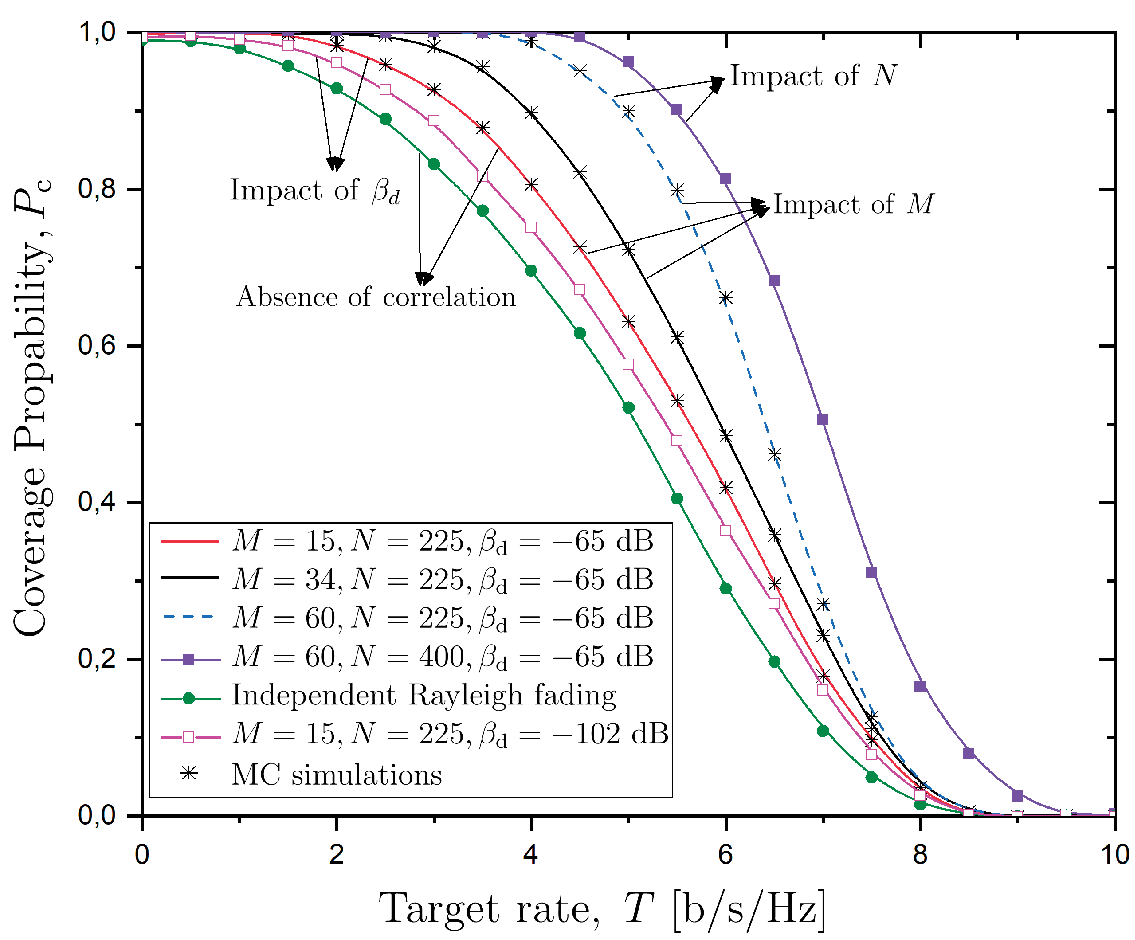} \vspace*{-0.2cm}
			\\$(b)$
			\vspace*{-0.2cm}
		\end{minipage}
		\begin{minipage}{0.33\textwidth}
			\centering
			\includegraphics[trim=0cm -0.20cm 0cm 0.2cm, clip=true, width=2.2in]{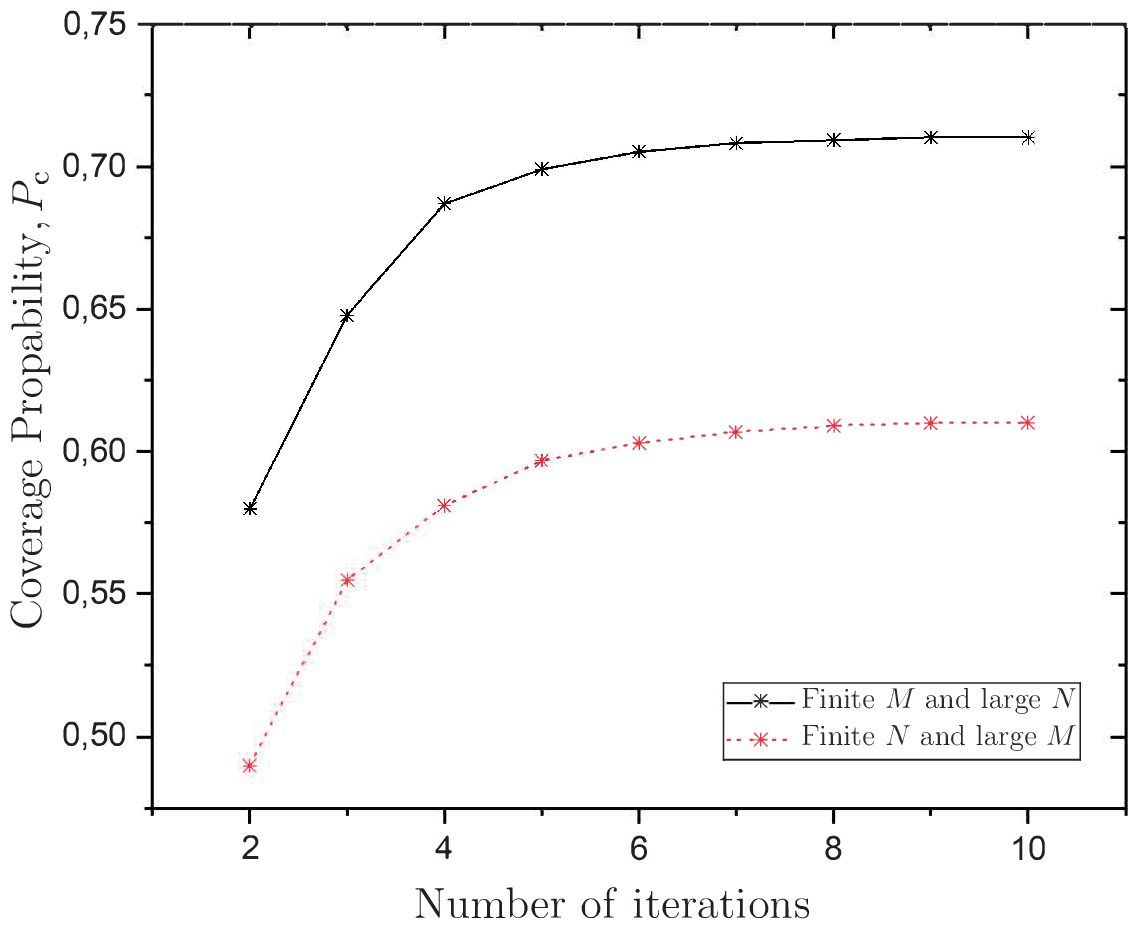}\vspace*{-0.2cm}\\$(c)$
			\vspace*{-0.2cm}
		\end{minipage}
		\caption{Coverage probability versus the target rate $ T $ (analytical results and MC simulations) of a SISO system with correlated Rayleigh fading assisted by $(a)$ $ M $ IRSs each having a large number of elements ($ N \to \infty $) and $(b)$ a large number of IRSs ($ M\to \infty $) each having $ N $ elements; $(c)$ Coverage probability versus the number of iterations for both cases.}
		\label{Fig1}
		\vspace{-0.65cm}
	\end{figure*}
	
	In Fig. \ref{Fig1}.(b), we depict the coverage probability versus the target rate by accounting for a large number of surfaces, i.e., we shed light on the setting referring to Proposition \ref{GeneralPDF1}. Herein, we notice that the coverage is improved as the number of IRSs increases. Similarly, if we increase the number of elements per IRS, we observe a further improvement. Notably, a comparison between Figs. \ref{Fig1}.(a) and \ref{Fig1}.(b) reveal that an analogous increase concerning the number of IRSs results in a larger improvement of the coverage compared to increasing the number of elements per IRS. Specifically, in Fig. \ref{Fig1}.(b), it is shown that when $ M $ increases from $15 $ to $ 34 $ ($ 225\% $ increase), $ P_{\mathrm{c}} $ starts decreasing from full coverage when $ T=1.47~\mathrm{dB}$ and $ 2.2~\mathrm{dB}$, respectively. On the other hand, in Fig. \ref{Fig1}.(a), we observe that for a similar increase concerning the number of elements per IRS, i.e., when $ N $ increases from $ 100 $ to $ 225 $, the coverage is much lower.
	
	In Fig. \ref{Fig1}.(c), we show the performance (convergence behaviour) achieved by Algorithm 1 versus the number of iterations for both cases studied in this work, i.e, i) finite $ M $ and large $ N $ (black solid line) and ii) finite $ N $ and large $ M $ (red dotted line). We observe that the coverage probability, provided by this algorithm, takes quickly its optimal value (converges) with the number of iterations in both cases (in less than $ 10 $ iteration), which proves the robustness of this method. 
	\section{Conclusion} \label{Conclusion} 
	In this paper, we derived the coverage probability of a SISO system assisted with multiple IRSs under the unavoidable conditions of correlated Rayleigh fading. We considered two distinct scenarios: a finite multitude of large IRSs and a large number of finite IRSs. Especially, we managed to derive and optimize the coverage probability with respect to the phase shifts of the IRS elements in both cases. The results enabled us to show that it is more beneficial to increase the number of IRSs instead of increasing their elements. Future works on coverage of distributed IRSs should take into account the design of multi-user transmission with multiple antennas, and possibly, under Rician fading conditions.
	
	\begin{appendices}
		\section{Proof of Lemma~\ref{deriv1}}\label{ArbitraryPDFProof2}
		First, by applying the chain rule, we obtain 
		\begin{align}
			\pdv{	{P_{\mathrm{c}}}}{\bs_{m,i}^{*}}&=\pdv{	{P_{\mathrm{c}}}}{B_{i}}\pdv{	B_{i}}{\bs_{m,i}^{*}},\label{deriv2}
		\end{align}
		where $ i=M$ and $ i=N$ correspond to Propositions \ref{GeneralPDF} and Propositions \ref{GeneralPDF1}, respectively.
		The first-order derivative in \eqref{deriv2} becomes
		\begin{align}
			\pdv{	{P_{\mathrm{c}}}}{	B_{i}}&=\frac{1}{\beta_{\mathrm{d}}}	{P_{\mathrm{c}}}.\label{deriv3}
		\end{align}
		By taking into account the expression of $ 	B_{M}$, we have
		\begin{align}
			\pdv{	B_{M}}{\bs_{m,i}^{*}}&= \beta_{m}\pdv{\left(\diag\left( \bR_{m,1}\bPhi_{m}\bR_{m,2}\right)\right)^{\T}\bs_{m,i}^{*}}{\bs_{m,i}^{*}}\\
			&= \beta_{m}\diag\left(\bR_{m,1}\bPhi_{m} \bR_{m,2}\right),\label{deriv4}
		\end{align}
		where we have used the property $ \tr\left(\bA \diag(\bs_{m})\right)=\left(\diag(A)\right)^{\T}\bs_{m} $.
		In the case of $ \pdv{	B_{N}}{\bs_{m,i}^{*}} $, we obtain
		\begin{align}
			\pdv{	B_{N}}{\bs_{m,i}^{*}}&=	\displaystyle\!\sum^{N}_{p=1}\!\pdv{\left(\diag\left(\bQ_{np,1}\bPsi_{n} \bQ_{np,2}\right)\right)^{\T}\bs_{m,i}^{*}}{\bs_{m,i}^{*}}\\
			&=	\displaystyle\!\sum^{N}_{p=1}\diag \left( \bQ_{np,1}\bPsi_{n} \bQ_{np,2}\right)\nn\\
			&=\beta_{m}	\displaystyle\!\sum^{N}_{p=1}\begin{bmatrix}r_{1p,m}^{1} r_{1p,m}^{2}\phi_{m,1} \\\vdots \\ r_{Np,m}^{1} r_{Np,m}^{2}\phi_{m,N}\end{bmatrix}\!.\label{deriv5}
		\end{align}

		By substituting \eqref{deriv4} or \eqref{deriv5} together with \eqref{deriv3} into \eqref{deriv2}, we obtain the desired results.

	\end{appendices}
	\bibliographystyle{IEEEtran}
	
	\bibliography{IEEEabrv,mybib}

\begin{thebibliography}{10}
\providecommand{\url}[1]{#1}
\csname url@samestyle\endcsname
\providecommand{\newblock}{\relax}
\providecommand{\bibinfo}[2]{#2}
\providecommand{\BIBentrySTDinterwordspacing}{\spaceskip=0pt\relax}
\providecommand{\BIBentryALTinterwordstretchfactor}{4}
\providecommand{\BIBentryALTinterwordspacing}{\spaceskip=\fontdimen2\font plus
\BIBentryALTinterwordstretchfactor\fontdimen3\font minus
  \fontdimen4\font\relax}
\providecommand{\BIBforeignlanguage}[2]{{%
\expandafter\ifx\csname l@#1\endcsname\relax
\typeout{** WARNING: IEEEtran.bst: No hyphenation pattern has been}%
\typeout{** loaded for the language `#1'. Using the pattern for}%
\typeout{** the default language instead.}%
\else
\language=\csname l@#1\endcsname
\fi
#2}}
\providecommand{\BIBdecl}{\relax}
\BIBdecl

\bibitem{Wu2020}
Q.~{Wu} and R.~{Zhang}, ``Towards smart and reconfigurable environment:
  Intelligent reflecting surface aided wireless network,'' \emph{IEEE Commun.
  Mag.}, vol.~58, no.~1, pp. 106--112.

\bibitem{Basar2019}
E.~Basar \emph{et~al.}, ``Wireless communications through reconfigurable
  intelligent surfaces,'' \emph{IEEE Access}, vol.~7, pp. 116\,753--116\,773,
  2019.

\bibitem{Huang2019a}
C.~Huang \emph{et~al.}, ``Reconfigurable intelligent surfaces for energy
  efficiency in wireless communication,'' \emph{IEEE Trans. Wireless Commun.},
  vol.~18, no.~8, pp. 4157--4170, 2019.

\bibitem{Bjoernson2019b}
E.~Bj{\"o}rnson, {\"O}.~{\"O}zdogan, and E.~G. Larsson, ``Intelligent
  reflecting surface versus decode-and-forward: {How} large surfaces are needed
  to beat relaying?'' vol.~9, no.~2, pp. 244--248.

\bibitem{Pan2020}
C.~{Pan} \emph{et~al.}, ``Multicell {MIMO} communications relying on
  intelligent reflecting surfaces,'' \emph{IEEE Trans. Wireless Commun.},
  vol.~19, no.~8, pp. 5218--5233, 2020.

\bibitem{Kammoun2020}
Q.~U.~A. {Nadeem} \emph{et~al.}, ``Asymptotic max-min {SINR} analysis of
  reconfigurable intelligent surface assisted {MISO} systems,'' \emph{IEEE
  Trans. Wireless Commun.}, vol.~19, no.~12, pp. 7748--7764, 2020.

\bibitem{Elbir2020}
A.~M. {Elbir} \emph{et~al.}, ``Deep channel learning for large intelligent
  surfaces aided {mm-Wave} massive {MIMO} systems,'' \emph{IEEE Wireless
  Commun. Lett.}, vol.~9, no.~9, pp. 1447--1451, 2020.

\bibitem{Najafi2020}
M.~{Najafi} \emph{et~al.}, ``Physics-based modeling and scalable optimization
  of large intelligent reflecting surfaces,'' \emph{IEEE Trans. Commun.}, pp.
  1--1, 2020.

\bibitem{Guo2020}
C.~{Guo} \emph{et~al.}, ``Outage probability analysis and minimization in
  intelligent reflecting surface-assisted {MISO} systems,'' \emph{IEEE Commun.
  Lett.}, vol.~24, no.~7, pp. 1563--1567, 2020.

\bibitem{Yang2020}
L.~{Yang} \emph{et~al.}, ``Coverage, probability of {SNR} gain, and {DOR}
  analysis of {RIS}-aided communication systems,'' \emph{IEEE Wireless Commun.
  Lett.}, vol.~9, no.~8, pp. 1268--1272, 2020.

\bibitem{Gao2020}
Y.~{Gao} \emph{et~al.}, ``Distributed {IRS} with statistical passive
  beamforming for {MISO} communications,'' \emph{IEEE Wireless Commun. Lett.},
  pp. 1--1, 2020.

\bibitem{Zhang2019a}
Z.~Zhang \emph{et~al.}, ``Analysis and optimization of outage probability in
  multi-intelligent reflecting surface-assisted systems,'' \emph{arXiv preprint
  arXiv:1909.02193}, 2019.

\bibitem{Sun2020}
S.~{Sun} \emph{et~al.}, ``Towards reconfigurable intelligent surfaces powered
  green wireless networks,'' in \emph{IEEE Wireless Communications and
  Networking Conference (WCNC)}, 2020, pp. 1--6.

\bibitem{Zhang2020}
Z.~Zhang and L.~Dai, ``Capacity improvement in wideband reconfigurable
  intelligent surface-aided cell-free network,'' in \emph{IEEE 21st
  International Workshop on Signal Processing Advances in Wireless
  Communications (SPAWC)}, 2020, pp. 1--5.

\bibitem{Bjoernson2020}
E.~{Bj{\"o}rnson} and L.~{Sanguinetti}, ``Rayleigh fading modeling and channel
  hardening for reconfigurable intelligent surfaces,'' \emph{IEEE Wireless
  Commun. Lett.}, vol.~10, no.~4, pp. 830--834, 2021.

\bibitem{Neumann2018}
D.~Neumann, M.~Joham, and W.~Utschick, ``Covariance matrix estimation in
  massive {MIMO},'' \emph{IEEE Signal Process. Lett.}, vol.~25, no.~6, pp.
  863--867, 2018.

\bibitem{Papazafeiropoulos2015a}
A.~K. Papazafeiropoulos and T.~Ratnarajah, ``Deterministic equivalent
  performance analysis of time-varying massive {MIMO} systems,'' \emph{IEEE
  Trans. Wireless Commun.}, vol.~14, no.~10, pp. 5795--5809, 2015.

\bibitem{Papazafeiropoulos2016}
A.~K. Papazafeiropoulos, ``Impact of general channel aging conditions on the
  downlink performance of massive {MIMO},'' \emph{IEEE Trans. Veh. Tech.},
  vol.~66, no.~2, pp. 1428--1442, Feb 2017.

\bibitem{Boyd2004}
S.~Boyd, S.~P. Boyd, and L.~Vandenberghe, \emph{Convex optimization}.\hskip 1em
  plus 0.5em minus 0.4em\relax Cambridge university press, 2004.

\end{thebibliography}
\end{document}